\documentclass[graybox]{SNmult}

\usepackage{type1cm}        
\usepackage{makeidx}         
\usepackage{graphicx}        
\usepackage{multicol}        
\usepackage[bottom]{footmisc}

\usepackage{newtxtext}       %
\usepackage[varvw]{newtxmath}       

\makeindex             

\newtheorem{prop}{Proposition}
\newtheorem{fact}{Fact}
\theoremstyle{definition}
\newtheorem*{rem}{Remark}
\newtheorem*{rems}{Remarks}
\usepackage{bm}
\usepackage{cite}

\begin{document}
\title*{What aggregation rules can be classified as logical concepts?}
\author{Nikolay L. Polyakov\orcidID{0000-0003-3124-8706} 
}
\institute{Nikolay L. Polyakov \at HSE University, 11 Pokrovsky Boulevard, Moscow \email{npolyakov@hse.ru}
}
%
%
\maketitle
\abstract*{In this paper, we study aggregation rules with nontrivial symmetric classes of invariant sets (restricted domains), assuming that they, unlike others, have a logical nature. In the simplest case, we provide a complete classification of such rules. Our primary tools are methods of universal algebra and the theory of closed classes of discrete functions.
}

\abstract{In this paper, we study aggregation rules with nontrivial symmetric classes of invariant sets (restricted domains), assuming that they, unlike others, have a logical nature. In the simplest case, we provide a complete classification of such rules. Our primary tools are methods of universal algebra and the theory of closed classes of discrete functions.
}

\section{Introduction} 

In this paper, we study aggregation rules with symmetric (closed under isomorphisms) classes of invariant domains. The primary motivation for this research is the belief that these rules stand out from the rest due to their \textit{logical} nature. Although the division between the logical and the non-logical is rather arbitrary, Alfred Tarski's idea can serve as a guiding thread: ``...a notion (individual, set, function, etc.) based on a fundamental universe of discourse is said to be logical if and only if it is carried onto itself by each one-one function whose domain and range both coincide with the entire universe of discourse'', \cite{Tarski}. A detailed discussion and development of this approach can be found in \cite{Sher1,Sher2}; see also \cite{Dragalina}. Extending this idea to the field of social choice theory, we will consider classes $\mathscr D$ of sets of individual preferences (or choices of variants) and aggregation rules $f$ that satisfy the following conditions: (a) the class $\mathscr D$ is closed under natural isomorphisms, and (b) every set $D\in \mathscr D$ is closed under the rule $f$. It can be noted that such a situation allows for a logical formalism in which the propositions are individual preferences, and $f$ plays the role of an inference rule. Similar constructions are discussed, for example, in \cite{Endriss,Porello}. Given the very general nature of the problem, we will be forced to make additional efforts to eliminate trivial cases. To some extent, this corresponds to eliminating empty and inconsistent logics, which, respectively, derive no propositions or derive all propositions at once. Our main tool will be the universal algebraic approach proposed by Shelah in \cite{Shelah}, which we believe has great promise. In developing this approach, we utilize the fundamental results of the theory of discrete functions, also known as functions of $k$-valued logic. Various connections between these scientific fields have been repeatedly noted in the literature; see, for example, \cite{Murakami} or \cite{Ovchinnikov}.

\section{General design}
First, we describe a general model of individual choice of options and collective decision-making. Typically (see, e.g., \cite{Aleskerov1,Aleskerov2,Aleskerov3,Aleskerov 2024} or \cite{Hand}), individual choice is formalized using a \textit{choice function} $c$ defined on the set $\mathscr P(A)$ or some natural subset thereof, where $A$ is a set of options (alternatives), $\mathscr P(A)$ is the set of its subsets, and the choice function $c$ is defined by the condition $c(X)\subseteq X$ (or $c(X)\in X$ if the study is restricted to the case of a single choice, as, e.g., in \cite{Shelah}). Models of \textit{individual preferences} can generally be reformulated in the language of choice functions. However, if we focus on aggregation rules, we can consider more general situations in which the domain of the choice function $c$ is a set of alternatives together with some additional structure. For example, an agent may choose an option taking into account both their own preferences and the results of previous votes (related models were discussed in \cite{Polyakov_Shamolin_nonlocal}) or some other additional information about the options. Within this framework, individual choice can be formalized using a function $c$ defined on graphs with vertices in $A$, first-order models with domains $X\subseteq A$, topologies on the set $A$ or its subsets, multisubsets or fuzzy subsets of $A$, sequences of elements of $A$, etc. In the most general case, it suffices for us to assume that the model of individual choice (if we wish to view it as a logical object) is closed under isomorphisms generated by permutations of the set of alternatives. This leads to the following definition.

\begin{definition}
	A \textit{model of individual choice} is a quadruple $(A, B, C, \ast)$, where $A$ and $B$ are nonempty sets (of \textit{options} and \textit{conditions}\footnote{Or \textit{choice contexts}, using the terminology of \cite{Aleskerov 2024}}, respectively), $\ast$ is an embedding of the permutation group of $A$ into the permutation group of $B$, and (the set of \textit{choice functions}) $C$ is a nonempty set of functions from $B$ to $\mathscr P(A)$ (or to $A$ in the case of single choice), such that for any permutation $\sigma$ of $A$ and any function $c\in C$, the set $C$ contains the function $c_\sigma$ defined by: for all $b\in B$
	$$
	c_\sigma (b)=\sigma^{-1}(  c(\sigma^\ast(b))).
	$$
	Here we write $\sigma^\ast$ instead of $\ast(\sigma)$ and (for the case of multiple choice) set $\sigma^{-1}(Y)={\sigma^{-1}(y): y\in Y}$ for all $Y\subseteq A$.
\end{definition}

The process of collective decision-making is formalized by an \textit{aggregation rule}. We will understand an aggregation rule as any function $f: C^n\to C$ satisfying the unanimity condition:

$$
f(  c,   c, \ldots,   c)=  c
$$
(no matter how general the approach to collective decision-making procedures, it is hard to find any argument against the unanimity condition).

\par The structures obtained in this way can be considered as logical concepts; however, a number of known results suggest that under certain natural assumptions, we have no nontrivial examples. The famous Arrow theorem \cite{Arrow} implies that for $3\leq |A|<\omega$, the set $C$ of \textit{rational} choice functions admits only dictatorial rules within the class of local aggregation rules. From the works \cite{Shelah, Polyakov1} it follows that this "impossibility principle" extends quite widely. Namely, except for a few exotic cases, for $3\leq |A|<\omega$ the model of local aggregation of single-choice functions on the set of $r$-element subsets of $A$ leads to exactly two situations: either the rule $f$ is dictatorial, or the set $C$ consists of all admissible choice functions (here one can see an analogy with two degenerate cases in logic: the empty and the inconsistent theory).

\par The above observations lead us to the following concept. For lack of a better alternative, we will assume that the model of individual choice $(A, B, C, \ast)$ is maximally broad (and therefore "trivial"), i.e., it contains, among others, all "absurd" (unacceptable for various reasons) choice functions that can be obtained via aggregation, and we will focus on aggregation rules $f: C^n\to C$ that \textit{preserve} some natural set of "reasonable" choice functions. Many positive results (possibility theorems) are known in this area. A description of the class of invariants (\textit{restricted domains}) for the majority rule can be found in the works of Sen \cite{Sen}, Kaneko \cite{Kaneko}, and others. The specific class of single-peaked domains (see \cite{Mulin}) has been studied in detail in \cite{Karpov1, Karpov2, Karpov3}. Restricted (invariant) domains for the \textit{nonlocal} aggregation procedure according to the level $1$ consensus rule can be derived from \cite{level 1}, etc.

\par With this approach, in what case can an aggregation rule $f: C^n\to C$ be called \textit{logical}? We believe that the minimal condition for this is the existence of a class $\mathscr D$ of nontrivial invariant sets $D\subseteq C$ for $f$ that is closed under permutations of the set of options, i.e., together with each set $D$ it also contains the set $\{d_\sigma: d\in D\}$ for every permutation $\sigma$ of $A$. Intuitively, this means that a logical aggregation rule effectively "works" on some nontrivial restricted domains $D$, and this fact is due only to the general structure of the domain $D$ (without hidden "privileges" for any specific options). It is clear that the condition of \textit{nontriviality} of a set $D\subseteq C$ cannot be precisely formalized in the general case. However, the classes $D$ consisting of all choice functions $c\in C$ that coincide on some set $X\subseteq B$ should certainly be considered trivial, since they are preserved by every aggregation rule $f: C^n\to C$ (in particular, all singletons $\{c\}$, $c\in C$, and the set $C$ itself are trivial). We will see below that in some simple cases this is sufficient to obtain a characterization of logical aggregation rules. All of the above leads us to the following definitions.

\par Let $\mathfrak C=(A,  B, C, \ast)$ be a model of individual choice, and $\mathcal F$ be the set of all possible aggregation rules for $\mathfrak C$, i.e.
$$
\mathcal F=\bigcup\limits_{n<\omega}C^{C^n}.
$$
\begin{definition}
	We say that an aggregation rule $f\in \mathcal F$ \textit{preserves} a set \mbox{$  D\subseteq   C$}, and a set $  D$ is \textit{preserved} by $f$ (or is an \textit{invariant set} for $f$) if 
	$
	f( c_1, c_2, \ldots, c_n)\in  D
	$
	for all $ c_1,  c_2,\ldots,  c_n\in   D$. For any aggregation rule $f$, the set of all sets \mbox{$  D\subseteq   C$} that are invariant for $f$ is denoted by $\mathrm{Inv}\, f$. For any set $\mathcal F\subseteq \mathcal F$  we denote $\mathrm{Inv}\, \mathcal F =\bigcap\limits_{f\in \mathcal F}\mathrm{Inv}\, f$. Symmetrically,~$\mathrm{Pres}\, D$ is the set of all aggregation rules $f$ that preserve a set $D\subseteq C$, and $\mathrm{Pres}\, \mathscr D=\bigcap\limits_{D\in \mathscr D}\mathrm{Pres}\, D$ for any $\mathscr D\subseteq \mathscr P(C)$. 
\end{definition}
The set of all permutations of $A$ is denoted by $S_A$.
\begin{definition}
	A set $\mathscr D\subseteq \mathscr P(C)$ is \textit{symmetric} (with respect to the embedding $\ast$) if for any permutation $\sigma \in S_A$ and set $D\in \mathscr D$ the set $D_\sigma\rightleftharpoons	
	\{d_\sigma: d\in D\}$ belongs to $\mathscr D$. A set $D\subseteq C$ is \textit{symmetric} if the singleton $\{D\}$ is symmetric.
\end{definition}
\begin{definition}
	A set $D\subseteq C$ is \textit{trivial} if  $C$ is empty or $D=\{c\in C: c\restriction_X=d\restriction_X\}$ for some $X\subseteq B$ and $d\in C$. 
	A set $\mathscr D\subseteq \mathscr P(C)$ is \textit{trivial} if it consists only of trivial sets $D\subseteq C$. 
\end{definition}

Thus, we consider the main indicator of the "logicality" of an aggregation rule $f$ to be the existence of a symmetric nontrivial set $\mathscr D\subseteq \mathrm{Inv}\, f$. 

\par To study logical aggregation rules, one can develop a universal-algebraic approach. Indeed, it is easy to verify that the following statement holds.

\par For any set $X$, a set of functions $\mathcal F\subseteq \bigcup_{n<\omega} X^{X^n}$ is called a \textit{clone} on $X$ (see, e.g., \cite{Lau}) if it contains all projections, i.e., the functions defined by the identities
$$
f(x_1, x_2, \ldots, x_n)=x_i
$$
(for some fixed $i$, $1\leq i\leq n$),
and is closed under composition, i.e., together with functions $f:X^n\to X$ and $f_1, f_2, \ldots, f_n: X^m\to X$ it contains the function $f(f_1, f_2, \ldots, f_n): X^m\to X$ defined by the identity
$$
f( f_1, f_2, \ldots, f_n)(\bm x)= f( f_1(\bm x), f_2(\bm x), \ldots, f_n(\bm x)).
$$
A clone $\mathcal F$ on $X$ is called \textit{idempotent} if
$$
f(x, x, \ldots, x)=x
$$
for every function $f\in \mathcal F$ and $x\in X$.

\par For each aggregation rule $f:C^n\to C$ and permutation $\sigma\in S_A$, define the function $f_\sigma: C^n\to C$ by the rule
$$
f_\sigma(c^1, c^2, \ldots, c^n)=(f(c^1_\sigma, c^2_\sigma, \ldots, c^n_\sigma))_{\sigma^{-1}}
$$
for all $c^1, c^2, \ldots, c^n\in C$. A set of aggregation rules is called \textit{symmetric} if, together with each function $f$, it contains all functions $f_\sigma$, $\sigma\in S_A$.

\begin{prop}\label{Prop 1}
	Let $\mathfrak C=(A,  B, C, \ast)$ be a model of individual choice and $\mathcal F$ the set of all possible aggregation rules for $\mathfrak C$. Then for any set $\mathscr D\subseteq \mathscr P(C)$
	\begin{enumerate}
		\item the set $\mathrm{Pres}\,\mathscr D$ is an idempotent clone on $C$,
		\item if $\mathscr D$ is symmetric, then $\mathrm{Pres}\,\mathscr D$ is symmetric.
	\end{enumerate}
\end{prop}
\begin{proof}
	By a routine verification.
\end{proof}
\begin{rem}
	We can also notice (this is not used in this work) that the pair $(\mathrm{Inv}, \mathrm{Pres})$ is an antitone Galois connection between the Boolean lattices $\mathscr P(\mathcal F)$ and $\mathscr P(\mathscr P(B))$, and all Galois-closed sets $\mathcal G\subseteq \mathcal F$ are clones. The connection $(\mathrm{Inv}, \mathrm{Pres})$ is closely related to the classical connection $(\mathrm{Inv}, \mathrm{Pol})$, well known in the theory of closed functional classes (see, e.g., \cite{Lau} or \cite{Polyakov_Shamolin_clon}), but is not reducible to it.
\end{rem}
The conditions of Proposition \ref{Prop 1} impose rather strict restrictions. Below, we assume that, in the simplest case, they allow for a complete description of \textit{local} aggregation rules of a logical nature.

\section{The simplest case: Logicality $+$ Locality $\approx$ Neutrality}
\par In this section we will consider the model of individual choice $\mathfrak C_2(A)=(A, B, C, \ast)$, where $B$ is the set $[A]^2$ of all two-element subsets of $A$, $\ast$ is the natural embedding of the group $S_A$ into the group $S_B$ (i.e., the embedding $\ast$ given by $\sigma^\ast({x, y})=\{\sigma(x), \sigma(y)\}$ for all distinct $x, y\in A$), and $C$ is the set of all choice functions $c:B\to A$ (i.e., functions $c$ satisfying $c(b)\in b$ for every $b\in B$). In a sense, $\mathfrak C_2(A)$ is the simplest model of individual choice. We will provide a convenient description of local logical aggregation rules for $\mathfrak C_2(A)$ in the case $5\leq |A|<\omega$. In this section we will discuss only models $\mathfrak C_2(A)=(A, B, C, \ast)$; the meaning of the symbols $A$, $B$, $C$ is assumed to be fixed throughout the section.
\par We begin with a series of standard definitions, various modifications of which have been known since classical works on social choice theory; see, for example, \cite{Arrow, Arrow1} or \cite{Aleskerov1,Aleskerov2,Aleskerov3}.
\par For any connected asymmetric binary relation $P$ on $A$ the choice function $c_P\in C$ is defined as follows:
$$
c_P(\{a_1, a_2\})=a_1\Leftrightarrow a_1\, P\, a_2
$$
for all distinct $a_1, a_2\in A$. Obviously, the mapping $P\mapsto c_P$ is a bijection from the set of all connected asymmetric relations on $A$ onto $C$. A choice function $c\in C$ is called \textit{rational} if 
$$
c=c_\prec
$$
for some linear order $\prec$ on $A$. A set $D\subset C$ is called \textit{rational} if it consists only of rational functions $c\in C$. A class $\mathscr D\subseteq \mathscr P(C)$ is called \textit{rational} if it consists only of rational sets $D\subseteq C$.

\par  An aggregation rule $f: C^n\to C$ (for $\mathfrak C_2(A)$) is called
\begin{enumerate}
	\item \textit{local} if $$(c_1(b)=d_1(b)\wedge \ldots\wedge c_n(b)=d_n(b))\Rightarrow f(c_1, c_2, \ldots, c_n)(b)=f(d_1, d_2, \ldots, d_n)(b)$$ 
	for all $c_1$, $c_2$,~$\ldots$, $c_n$, $d_1$, $d_2$, $\ldots$, $d_n\in C$ and $b\in B$,
	\item  \textit{neutral} if for each permutation $\sigma\in S_A$ we have: $$f_\sigma=f,$$
	\item \textit{dictatorial} (or a \textit{dictatorship}) if it is a projection, i.e. for some $i$, $1\leq i\leq n$, we have
	$$
	f(c_1, c_2, \ldots, c_n)=c_i
	$$
	for all $c_1$, $c_2$,~$\ldots$, $c_n\in C$.
\end{enumerate}
\par A \textit{decisive coalition} on a set (of agents) $E_n\leftrightharpoons\{1, 2, \ldots, n\}$ is a set $\mathscr K\subseteq \mathscr P(E_n)$ such that for any set $X\subseteq E_n$, exactly one of the sets $X$ and $E_n\setminus X$ belongs to $\mathscr K$. It is well known that every local and neutral aggregation rule $f:C^n\to C$ can be defined using a decisive coalition $\mathscr K_f\subseteq \mathscr P(E_n)$ as follows: for all $c_1$, $c_2$,~$\ldots$, $c_n\in C$ and $b=\{a_1, a_2\}\in B$
$$
f(c_1, c_2, \ldots, c_n)(b)=a_1\Leftrightarrow \{i\in E_n: c_i(b)=a_1\}\in \mathscr K_f.
$$
\par All local and neutral aggregation rules can be classified as logical due to the following observation.
\begin{prop}\label{PropNeutrality}
	Every local and neutral aggregation rule $f:C^n\to C$ for $\mathfrak C_2(A)$ preserves any two-element set $D\subseteq C$.
\end{prop}
\begin{proof}
	Let $f:C^n\to C$ be a local and neutral aggregation rule for $\mathfrak C_2(A)$, $D\subseteq C$, $|D|=2$, and $(d^1, d^2, \ldots, d^n)\in D^n$. Choose $c^1, c^2\in D$ and $b\in [A]^2$ such that $c^2(b)\neq c^1(b)=f(d^1, d^2, \ldots, d^n)(b)$. It is enough to show that $f(d^1, d^2, \ldots, d^n)(b')=c^1(b')$ for all $b'\in [A]^2$. If $c^1(b')=c^2(b')$, this follows immediately from the conditions of locality and unanimity. Otherwise, consider a permutation $\sigma\in S_A$ for which $\sigma(c^1(b))=c^1(b')$ and $\sigma(c^2(b))=c^2(b')$. Since $f$ is local, we have:
	$
	f(d^1_\sigma, d^2_\sigma, \ldots, d^n_\sigma)(b)= f(d^1, d^2, \ldots, d^n)(b)=c^1(b).
	$.
	Then since $f$ is neutral, we have:
	$$
	f(d^1, d^2, \ldots, d^n)(b')=\sigma(f(d^1_\sigma, d^2_\sigma, \ldots, d^n_\sigma)(b))=\sigma(c^1(b))=c^1(b').
	$$
\end{proof}
We show that the converse is also true for $5\leq |A|<\omega$, and we provide a simple classification of all local and neutral aggregation rules.
\par For each set $A$, we define four special local and neutral aggregation rules $\delta, \nu, \lambda, \mu: C^3\to C$ with decisive coalitions:
\begin{align*}
	\mathscr K_\delta=&\,\,\{\{1\}, \{1,2\}, \{1,3\}, \{1,2,3\}\},\\
	\mathscr K_\nu=&\,\,\{\{2\}, \{3\}, \{2,3\}, \{1,2,3\}\},\\
	\mathscr K_\lambda=&\,\,\{\{1\}, \{2\}, \{3\}, \{1,2,3\}\},\\
	\mathscr K_\mu=&\,\,\{\{1,2\}, \{2, 3\}, \{1,3\}, \{1,2,3\}\},
\end{align*}
respectively.  Obviously, $\delta$ is a dictatorship, and $\mu$ is the majority rule.
\begin{definition}
	We will say that aggregation rules $f:C^n\to C$ and $g: C^m\to C$ (for $\mathfrak C_2(A)$) are \textit{invariantly equivalent} if they preserve the same sets $D\subseteq C$, i.e.
	$$
	\mathrm{Inv}\, f=\mathrm{Inv}\, g.
	$$
\end{definition}
\begin{theorem}\label{Th1}
	Let $5\leq |A|<\omega$. Then for any local aggregation rule $f: C^n\to C$ for $\mathfrak C_2(A)$ the following conditions are equivalent:
	\begin{enumerate}
		\item\label{Th1Item1} $f$ has a non-trivial symmetric rational class of invariant sets,
		\item\label{Th1Item2} $f$ has a non-trivial symmetric class of invariant sets (i.e., $f$ is logical),
		\item\label{Th1Item3} $f$ is neutral,
		\item\label{Th1Item4} $f$ is invariantly equivalent to one of the rules $\delta, \nu, \lambda, \mu$.
	\end{enumerate}
\end{theorem}

\begin{rems}~
	\begin{enumerate}
		\item Theorem \ref{Th1} can be obviously extended to the case of aggregation of individual preferences, which in general are represented by connected asymmetric relations. To do this, it is necessary to replace choice functions $c_P\in C$ with relations $P$ and adjust the definitions accordingly.
		\item Theorem \ref{Th1} lies between the impossibility and possibility theorems. The existence of a nontrivial symmetric class of invariants imposes very strong restrictions on the aggregation rule, but still does not reduce the situation to a dictatorship. Furthermore, Item (\ref{Th1Item4}) allows us to perceive Theorem \ref{Th1} as a kind of \textit{reduction theorem}. Some problems of finding restricted domains for an arbitrary local rule can be reduced to similar problems for rules $\nu$, $\lambda$, and $\mu$, or even for rules $\lambda$ and $\mu$, since it can be shown that $\mathrm{Inv}\, \nu=\mathrm{Inv}\, \mu\cap \mathrm{Inv}\, \lambda$. For example, Arrow's theorem (for $|A|\geq 5$) can be derived from Theorem \ref{Th1}. To do this, it suffices to verify that the class of rational choice functions is not preserved by either rule $\lambda$ or $\mu$.
		\item The statement of Theorem \ref{Th1} is not true for infinite sets $A$ and in the case $|A|\in \{2, 3, 4\}$. In the case $|A|=2$, we have $|C|=2$, and there are obviously no non-trivial sets $D\subseteq C$, while there are infinitely many non-neutral aggregation rules that can be obtained using non-self-dual Boolean functions. In the case $A=\{a_1, a_2, a_3\}$, consider the function $g: A^2\to A$ with the Cayley table:
						\begin{table}[h!]\centering
			\begin{tabular}{l|l|l|l}
				$g$ &  $a_1$ & $a_2$ & $a_3$\\
				\hline
				$a_1$ &  $a_1$ & $a_2$ & $a_1$\\
				\hline
				$a_2$ &  $a_2$ & $a_2$ & $a_3$\\
				\hline
				$a_3$ &  $a_1$ & $a_3$ & $a_3$\\
			\end{tabular}
		\end{table}	
	and define the aggregation rule $f: C^2\to C$ by letting
		$$
		f(c_1, c_2)(b)=g(c_1(b), c_2(b))
		$$
		for all $c_1, c_2\in C$ and $b\in [A]^2$. It is easy to see that the rule $f$ is well-defined, local and not neutral. Consider a set $D\subseteq C$ consisting of two choice functions $c_1, c_2$ of the ``rock-paper-scissors'' type:
		\begin{table}[h!]\centering
			\begin{tabular}{l|l|l}
				&  $c_1$ & $c_2$ \\
				\hline
				$\{a_1, a_2\}$ &  $a_1$ & $a_2$ \\
				\hline
				$\{a_2, a_3\}$ &  $a_2$ & $a_3$ \\
				\hline
				$\{a_1, a_3\}$ &  $a_3$ & $a_1$ \\
			\end{tabular}
		\end{table}
		It is easy to see that the set $D=\{c_1, c_2\}$ is symmetric and non-trivial. Moreover, 
		$$
		f(c_1, c_2)=f(c_2, c_1)=c_2,
		$$	
		so, $D$ is invariant for $f$. However, it can be noted that for the case $n=3$, Theorem \ref{Th1} remains true if we exclude Item (\ref{Th1Item2}).
		\par For $A=\{a_1, a_2, a_3, a_4\}$, consider the function $h: A^2\to A$ with the Cayley table:
		\begin{table}[h!]\centering
			\begin{tabular}{l|l|l|l|l}
				$h$ &  $a_1$ & $a_2$ & $a_3$ & $a_4$\\ 
				\hline
				$a_1$ &  $a_1$ & $a_2$ & $a_1$& $a_4$\\
				\hline
				$a_2$ &  $a_1$ & $a_2$ & $a_3$& $a_2$\\
				\hline
				$a_3$ &  $a_3$ & $a_2$ & $a_3$& $a_4$\\
				\hline
				$a_4$ &  $a_1$ & $a_4$ & $a_3$& $a_4$\\ 
			\end{tabular}
		\end{table}
		\\and define the aggregation rule $f: C^2\to C$ by letting
		$$
		f(c_1, c_2)(b)=h(c_1(b), c_2(b))
		$$
		for all $c_1, c_2\in C$ and $b\in [A]^2$. Again it is easy to see that the rule $f$ is well-defined, local and not neutral. Consider the class $\mathscr D\subseteq \mathscr P(C)$ of all two-element sets $\{c_{\prec_1}, c_{\prec_2}\}$ of choice functions $c_{\prec_1}$ and $c_{\prec_2}$ constructed according to linear orders 
		\begin{align*}
			& {\prec_1}: x_1\prec_1 x_2\prec_1 x_3 \prec_1 x_4,\\
			& {\prec_2}: x_2\prec_2 x_1\prec_2 x_4 \prec_2 x_3
		\end{align*}
		where $\{x_1, x_2, x_3, x_4\}=A$. Obviously, $\mathscr D$ is a symmetric non-trivial rational class. Moreover, the aggregation rule $f$ preserves each set $D\in \mathscr D$. To prove this, it suffices to note that $h$ acts as the first projection on pairs $(a_1, a_3)$, $(a_3, a_1)$, $(a_2, a_4)$ and $(a_4, a_2)$, and as the second projection on all other pairs. Therefore, for any pair of linear orders ${\prec_1}: x_1\prec_1 x_2\prec_1 x_3 \prec_1 x_4,$ and ${\prec_2}: x_2\prec_2 x_1\prec_2 x_4 \prec_2 x_3$ we have:
		$$
		f(c_{\prec_1}, c_{\prec_2})= c_{\prec_1}\text{ and } f(c_{\prec_2}, c_{\prec_1})= c_{\prec_2}
		$$
		if $\{x_1, x_2\}$ is either $\{a_1, a_3\}$ or $\{a_2, a_4\}$, and 
		$$
		f(c_{\prec_1}, c_{\prec_2})= c_{\prec_2}\text{ and } f(c_{\prec_2}, c_{\prec_1})= c_{\prec_1},
		$$
		otherwise.
		\par In general, the combinatorics of aggregation on a four-element set of alternatives turns out to be the most complex. It can be shown that no non-dictatorial local aggregation rule has a symmetric invariant $D\subseteq C$. Moreover, the model of individual choice from three-element sets of alternatives already provides such an example, see \cite{Polyakov1}. The equivalence of Items (\ref{Th1Item3}) and (\ref{Th1Item4}) of Theorem \ref{Th1} remains true in the case $|A|=4$.
		\par For the case of an infinite number of alternatives, see \cite{Fish}.
		\item It can be shown that rules $\delta$, $\lambda$, $\mu$ and $\nu$ are not pairwise invariantly equivalent, i.e., Item \ref{Th1Item4} cannot be improved.
		\item How natural are the aggregation rules that satisfy the conditions of Theorem \ref{Th1}? Local aggregation rules that are invariantly equivalent to the majority rule $\mu$ do not raise serious doubts. It is easy to see that these are precisely local and neutral rules with monotone decisive coalitions. The rule~$\nu$ (and some rules that are invariantly equivalent to it) can also be justified under certain circumstances. Suppose that agents $2$ and $3$ give a correct (or truthful) answer with probability $\alpha$, $1/2<\alpha<1$, and agent $1$ with probability $1-\alpha$. It can be shown that in this case $\nu$ maximizes the probability of a correct social decision (in the class of local and neutral aggregation rules). \par In contrast, the rule $\lambda$ appears quite paradoxical in the list of logical rules. It can be shown that all rules that are invariantly equivalent to it are structured as follows. The set of agents is divided into dummies and non-dummies, and the social group chooses an alternative $a$ if an odd number of non-dummy agents voted for it. Thus, the aggregation procedure resembles a children's \textit{counting-out game}. Despite this, our study provides no arguments against these rules. Moreover, the combinatorial properties of these rules are, in some ways, better than those of $\mu$ and $\nu$. For example, unlike them, rule $\lambda$ has non-trivial symmetric invariants $D\subseteq C$ (for $|A|\neq 2\pmod{4}$), see~\cite{Polyakov1}.
	\end{enumerate}
\end{rems}

\section{Universal algebraic approach and the proof of Theorem \ref{Th1}}
The following proof of Theorem \ref{Th1} uses universal algebraic methods, in particular, the theory of closed classes of discrete functions. Detailed information can be found in the monographs \cite{March, Lau, Poshel}. The key point of the proof is an application of Post's classification theorem, which can be found in the original paper \cite{Post} or (in a more modern presentation) in \cite{MarchBool, Zuk}.

\par The set of all local aggregation rules for $\mathfrak C_2(A)$ is denoted by $\mathrm{LOC}_2(A)$.  Local aggregation rules $f: C^n\to C$ admit representations by partial functions $\widehat f:A^n\to A$. For any set $A$ and natural number $n$ we denote
$$
A^n_2\leftrightharpoons \{(x_1, x_2, \ldots, x_n)\in A^n:|\{x_1, x_2, \ldots, x_n\}|\leq 2\}.
$$

\begin{definition} A function $g:A_2^n\to A$ is called an ($n$-ary) \textit{$2$-function} on $A$. \par An $n$-ary $2$-function $g$ on $A$ is called \textit{conservative} if  
	$$
	\bigvee\limits_{1\leq i\leq n} (g(x_1, x_2, \ldots, x_n)=x_i)
	$$
	for all $(x_1, x_2, \ldots, x_n)\in A_2^n$.
	\par For any permutation $\sigma\in S_A$ and $2$-function $g: A^n_2\to A$ the $2$-function $g_\sigma: A^n_2\to A$ is defined by letting $$
	g_\sigma(x_1, x_2, \ldots, x_n)=\sigma^{-1}(g(\sigma(x_1), \sigma(x_2), \ldots, \sigma(x_n)))
	$$
	for all $(x_1, x_2, \ldots, x_n)\in A^n_2$. \par A $2$-function $g: A^n_2\to A$ is \textit{self-dual} if $g_\sigma=g$ for every permutation $\sigma\in S_A$.
	\par Any set \mbox{$G$} of $2$-functions on $A$ (of all arities $n$, $1\leq n<\omega$), that is closed with respect to superposition and contains all projections, is called a \textit{$2$-clone} on $A$ (the concepts of superposition and projection for $2$-functions are understood in the natural sense). \par A $2$-clone $G$ on $A$ is 
	\begin{enumerate}
		\item \textit{conservative} if it contains only conservative $2$-functions,
		\item \textit{self-dual} if it contains only self-dual $2$-functions,
		\item \textit{symmetric} if for every function $g\in G$ and permutation $\sigma\in S_A$ it contains the function $g_\sigma$.
	\end{enumerate}
	The set (in fact, the symmetric $2$-clone) of all conservative $2$-functions on $A$ is denoted by $\mathrm{CON}_2(A)$. 
\end{definition}
We say that an ($n$-ary) 2-function $g:A^n_2\to A$ \textit{represents} an ($n$-ary) aggregation rule $f\in \mathrm{LOC}_2(A)$ if
$$
f(c_1, c_2, \ldots, c_n)(b)=g(c_1(b), c_2(b), \ldots, c_n(b)) 
$$
for all $c_1, c_2, \ldots, c_n\in C$ and $b\in B$. 
\begin{prop}\label{Prop3} Let $|A|\geq 2$. Then 
	\begin{enumerate}\label{Prop 2}
		\item for any aggregation rule $f\in\mathrm{LOC}_2(A)$ there is a unique $2$-function $\widehat f \in \mathrm{CON}_2(A)$ that represents $f$; the mapping $f\mapsto \widehat f$ is a bijection from $\mathrm{LOC}_2(A)$ onto $\mathrm{CON}_2(A)$;
		\item for any clone $\mathcal G\subseteq \mathrm{LOC}_2(A)$ the set $\widehat {\mathcal G}\leftrightharpoons\{\widehat f: f\in \mathcal G\}$ is a $2$-clone;
		\item if a clone $\mathcal G\subseteq \mathrm{LOC}_2(A)$ is symmetric, then the $2$-clone $\widehat {\mathcal G}$ is symmetric.
	\end{enumerate}
\end{prop}
\begin{proof}
	Let $f$ be an $n$-ary aggregation rule in $\mathrm{LOC}_2(A)$. To construct a $2$-function $\widehat f$, for each sequence $(x_1, x_2, \ldots, x_n)\in A^n_2$, choose a set $b\in [A]^2$ for which $\{x_1, x_2, \ldots, x_n\}\subseteq b$ (if among $x_1, x_2, \ldots, x_n$ there are distinct elements, $b$ is uniquely determined). Next, we choose functions $c_1, c_2, \ldots, c_n\in C$ for which $c_1(b)=x_1$, $c_2(b)=x_2$, $\ldots$, $c_n(b)=x_n$ and set $\widehat f(x_1, x_2, \ldots, x_n)= f(c_1, c_2, \ldots, c_n)(b)$. From the locality condition it follows that the function $\widehat f$ is well defined. The rest of the proposition is proved by routine verification.
\end{proof}
\par Proposition \ref{Prop 2} reduces the classification of local logical aggregation rules for $\mathfrak C_2(A)$ to the classification of symmetric conservative clones on $A$. A description of all such clones can be found in Post's classification theorem.
\par First, suppose $|A|=2$. Without loss of generality, $A=\{0,1\}$. Then, any ($2$-)clone on $A$ is a subclone of Post's class of all Boolean functions that preserve $\mathbf 0$ and $\mathbf 1$. In addition, a clone $H$ of Boolean functions is symmetric if and only if it is closed with respect to duality, i.e.
$$
h(x_1, x_2, \ldots, x_n)\in H\Rightarrow \overline{h(\overline x_1, \overline x_2, \ldots, \overline x_n)}\in H. 
$$
\par There are only six of these clones: $O_1$, $D_1$, $D_2$, $L_4$, $A_4$, $C_4$ (in Post's notation, see \cite{Post}). Generating functions for them are given in Table \ref{tab11}. Note that the self-dual clones from this list are exactly clones $O_1$, $D_1$, $D_2$, $L_4$.

\renewcommand{\arraystretch}{1.5}
\begin{table}[h!]\centering
	\caption{Clones of Boolean functions preserving $\textbf 0$ and $\textbf 1$ that are closed with respect to duality.}\label{tab11}
	\begin{tabular}{|l|l|l|l|l|l|}
		\hline
		Clone &  Generating functions & Clone & Generating functions& Clone & Generating functions\\
		\hline
		$O_1$ &  $x$ & $D_1$ & \begin{minipage}[b]{0.27\textwidth} \vspace{0.1cm}$\overline xy\vee \overline yz\vee yz$ or\\ $xy\vee yz\vee xz$, $x\oplus y\oplus z$\end{minipage}  & $D_2$ & $xy\vee yz\vee xz$\\
		\hline
		$L_4$ &  $x\oplus y\oplus z$ & $A_4$ & $xy, x\vee y$ & $C_4$ & $x\vee  y\overline z$\\
		\hline
	\end{tabular}
\end{table}
\begin{definition}
	Let $A$ be a set, $|A|\geq 2$, $n$ a natural number, and $H$ a set of Boolean functions. An $n$-ary $2$-function $g:A^n_2\to A$ is 
	\begin{enumerate}
		\item\textit{freely associated} with $H$ if for any set $b\in [A]^2$ there exist an $n$-ary function $h\in H$ and a bijection $\sigma: b\to \{0,1\}$ such that
		$$
		g(x_1, x_2, \ldots, x_n)=\sigma^{-1}h(\sigma(x_1), \sigma(x_2), \ldots, \sigma(x_n))
		$$
		for all elements $x_1, x_2, \ldots, x_n\in b$;
		\item\textit{dependently associated} with $H$ if there exists an $n$-ary function $h\in H$ such that
		$$
		g(x_1, x_2, \ldots, x_n)=\sigma^{-1}h(\sigma(x_1), \sigma(x_2), \ldots, \sigma(x_n))
		$$
		for any set $b\in [A]^2$, bijection $\sigma: b\to \{0,1\}$ and elements $x_1, x_2, \ldots, x_n\in b$.
	\end{enumerate}
	The set of all $2$-functions $g:A^n_2\to A$, $1\leq n<\omega$, that are freely (dependently) associated with $H$ is called a \textit{free} (respectively, \textit{dependent}) extension of $H$ to $A$ and is denoted by $H^\uparrow(A)$ (respectively, $H^\Uparrow(A)$). 
\end{definition}
For any conservative $2$-clone $G$ on $A$ and set $b\in [A]^2$ denote $G_b\leftrightharpoons \bigcup_{1\leq n<\omega}\{g\restriction_{b^n}: g\in G\cap A^{A^n}\}$. Obviously, $G_b$ is a clone on $b$.
\begin{prop}
	Let $A$ be a set of cardinality $|A|\geq 2$. Then
	\begin{enumerate}
		\item for any symmetric conservative clone $H$ of Boolean functions the set $H^\uparrow(A)$ is a symmetric conservative $2$-clone, and for any $b\in [A]^2$ the clone $G_b$ is naturally isomorphic to $H$, 
		\item for any self-dual conservative clone $H$ of Boolean functions, i.e., a clone $H\in \{O_1, D_1, D_2, L_4\}$, the set $H^\Uparrow(A)$ is a self-dual conservative $2$-clone, and for any $b\in [A]^2$ the clone $G_b$ is naturally isomorphic to $H$.
	\end{enumerate}
\end{prop}
\begin{proof}
	By a routine checking.
\end{proof}

\begin{theorem}\label{Th2_clone}
	Let $5\leq |A|<\omega$, and $G$ be a symmetric conservative $2$-clone on $A$. Then one of the following two conditions holds:
	\begin{enumerate}
		\item\label{cond1} $G$ is the free extension of one of the clones $O_1$, $D_1$, $D_2$, $L_4$, $A_4$, $C_4$ to $A$,
		\item\label{cond2} $G$ is the dependent extension of one of the clones $O_1$, $D_1$, $D_2$, $L_4$ to $A$.
	\end{enumerate}
\end{theorem}
\begin{proof} Since $G$ is symmetric, for any sets $b_1, b_2\in [A]^2$ the clones $G_{b_1}$ and $G_{b_2}$ are naturally isomorphic. In addition, each clone $G_b$, $b\in [A]^2$, is symmetric. Therefore, it is naturally isomorphic to some Post's class $H\in\{O_1, D_1, D_2, L_4, A_4, C_4\}$. 
	\par Case 1: there is a binary $2$-function $g\in G$ that is not a projection. 
	\par In this case the condition $|A|\geq 5$ implies the following fact, see \cite{Shelah}.
	\begin{fact}\label{fact1}
		For all $a_1, a_2, a_3, a_4\in A$ such that $\{a_1, a_2\}\neq \{a_3, a_4\}$, $x\in \{a_1, a_2\}$, and $y\in \{a_3, a_4\}$ there exists a binary function $g\in G$ for which $g(a_1, a_2)=x$ and $g(a_3, a_4)=y$. 
	\end{fact}
	Let us show that this implies $G=H^\uparrow(A)$. It is enough to prove that for any natural number $n$, $n$-ary function $h\in H^\uparrow(A)$, and set $Z\subseteq A^n_2$ there is a function $g\in G$ such that $g\restriction_Z=h\restriction_Z$. By induction on $|Z|$. The induction base ($|Z|=0$) is obvious. In addition, the induction step is obvious if $Z\subseteq b^n$ for some $b\in [A]^2$. In the remaining case, choose two different sets $b_1, b_2\in [A]^2$ and two sequences $\bm x, \bm y\in Z$ such that $\bm x\in (b_1\setminus b_2)^n$ and $\bm y\in (b_2\setminus b_1)^n$. By the induction hypothesis, there are functions $g_{\bm x}, g_{\bm y}\in G$ such that 
	$$
	g_{\bm x}\restriction_{Z\setminus \{\bm y\}}=h\restriction_{Z\setminus \{\bm y\}}\text{ and }\,\,g_{\bm y}\restriction_{Z\setminus \{\bm x\}}=h\restriction_{Z\setminus \{\bm x\}}.
	$$
	If for any function $h'\in H^\uparrow(A)$ we have
	$$
	h'\restriction_{Z\setminus \{\bm x, \bm y\}}= h\restriction_{Z\setminus \{\bm x, \bm y\}}\Rightarrow (h'(\bm x)=g_{\bm x}(\bm x)\vee h'(\bm y)=g_{\bm y}(\bm y)),
	$$
	all is proved. Otherwise, by the induction hypothesis, there are functions $g'_{\bm x}, g'_{\bm y}\in G$ such that 
	$$
	g_{\bm x}'\restriction_{Z\setminus \{\bm x, \bm y\}}=g'_{\bm y}\restriction_{Z\setminus \{\bm x, \bm y\}}=h\restriction_{Z\setminus \{\bm x, \bm y\}}, g_{\bm x}'(\bm x)\in b_1\setminus b_2, \text{ and } g_{\bm y}'(\bm x)\in b_2\setminus b_1.
	$$
	Using Fact \ref{fact1}, choose binary functions $h_1, h_2, h_3\in G$ for which
	\begin{itemize}
		\item $h_1(g_{\bm x}'(\bm x), g_{\bm y}(\bm x))=g_{\bm x}'(\bm x)$, $h_1(g_{\bm x}'(\bm y), g_{\bm y}(\bm y))=g_{\bm y}(\bm y)$, 
		\item $h_2(g_{\bm x}(\bm x), g_{\bm y}'(\bm x))=g_{\bm x}(\bm x)$, $h_2(g_{\bm x}(\bm y), g_{\bm y}'(\bm y))=g_{\bm y}'(\bm y)$,
		\item $h_3(g_{\bm x}'(\bm x), g_{\bm x}(\bm x))=g_{\bm x}(\bm x)$, $h_3(g_{\bm y}(\bm y), g_{\bm y}'(\bm y))=g_{\bm y}(\bm y)$.
	\end{itemize}
	It is easy to verify that we can put $$g=h_3(h_1(g'_{\bm x}, g_{\bm y}), h_2(g_{\bm x}, g_{\bm y}'))$$
	(for the case $b_1\cap b_2=\emptyset$ the reasoning may be simpler).
	\par Case 2: case 1 does not hold. Then all binary functions from $G$ are projections. Let us show that this implies $G=H^\Uparrow(A)$. Without loss of generality, assume $\{0,1\}\subseteq A$. It is enough to show that $G_{\{0,1\}}\in \{O_1, D_1, D_2, L_4\}$, and for all natural numbers $n\geq 1$, $n$-ary $2$-functions $g\in G$, sets $b\in [A]^2$, sequences $(x_1, x_2, \ldots, x_n)\in b^n$, and permutations $\sigma\in S_A$ mapping $b$ onto $\{0,1\}$, the following equality holds:
	\begin{equation}\label{eq1}
		g(x_1, x_2, \ldots, x_n)=\sigma^{-1}g\restriction_{\{0,1\}^n}(\sigma(x_1), \sigma(x_2), \ldots, \sigma(x_n)).
	\end{equation}
	Note that for each $n$-ary $2$-function $g\in G$ and sequence $(x_1, x_2, \ldots, x_n)\in A^n_2$ there is a binary $2$-function $g'\in G$ and indices $i,j$, $1\leq i,j\leq n$, for which 
	$$g(\sigma(x_1), \sigma(x_2), \ldots, \sigma(x_n))=g'(\sigma(x_i), \sigma(x_j))$$ for each permutation $\sigma\in S_A$ ($g'$ can be obtained from $g$ by identifying the variables). Since $g'$ is a projection, we have $g'(\sigma(x_i), \sigma(x_j))=\sigma(g'(x_i, x_j))$, and, so:
	\begin{equation}\label{eq2}
		g(\sigma(x_1), \sigma(x_2), \ldots, \sigma(x_n))=\sigma(g(x_1, x_2, \ldots, x_n)).
	\end{equation}
	If $(x_1, x_2, \ldots, x_n)\in b^n$, and $\sigma$ maps $b$ onto $\{0,1\}$, we immediately obtain~(\ref{eq1}). \par Let $h$ be an $n$-ary function from $G_{\{0,1\}}$. So, $h=g\restriction_{\{0,1\}}$ for some $g\in G$. Using (\ref{eq2}), for an arbitrary 
	sequence $(x_1, x_2, \ldots, x_n)\in \{0,1\}^n$ and non-identical permutation $\sigma\in S_{\{0,1\}}$, we have
	$$
	h(x_1, x_2, \ldots, x_n)=\sigma^{-1}(h(\sigma(x_1), \sigma(x_2), \ldots, \sigma(x_n))),
	$$
	i.e., $h$ is self-dual. Therefore, the clone $G_{\{0,1\}}$ consists of self-dual functions. Hence, $G_{\{0,1\}}$ is one of the clones $O_1$, $D_1$, $D_2$, $L_4$. 
\end{proof}
We can now give a proof of Theorem \ref{Th1} (omitting some details that can easily be reconstructed). 
\par The implication $(1)\Rightarrow (2)$ is obvious. 
\par In order to prove the implication $(2)\Rightarrow (3)$, we need Theorem \ref{Th2_clone}.  Let a local rule $f:C^n\to C$ for $\mathfrak C_2(A)$ preserve each set $D$ from a symmetric class $\mathscr D\subseteq \mathscr P(C)$. Consider the clone $\mathcal G$ generated by all rules $f_\sigma$, $\sigma\in S_A$. It is easy to check that $\mathcal G$ is a symmetric clone consisting of local aggregation rules, and $\mathscr D\subseteq \mathrm{Inv}\, f'$ for each $f'\in \mathcal G$. If condition~(\ref{cond1}) of Theorem \ref{Th2_clone} holds for the $2$-clone $G=\widehat{\mathcal G}$, then all rules from $\mathcal G$ (in particular, $f$) are neutral. It remains to prove that $\mathscr D$ is trivial in the opposite case. Let $D\in \mathscr D$, and let $Z$ be the set of all $b\in [A]^2$ such that $d(b)=d'(b)$ for all $d, d'\in D$. Choose a function $d\in D$ and put $D^+=\{c\in C: c\restriction_Z=d\restriction_Z\}$. It is enough to show that $D=D^+$. Let $d$ be an arbitrary function from $D^+$. Let us enumerate the set $[A]^2\setminus Z=\{b_1, b_2, \ldots, b_m\}$. For any number $i$, $1\leq i\leq m$, $D$ contains a function $d_i$ for which $d_i(b_i)=d(b_i)$. In the case under consideration, the clone $G$ satisfies condition~(\ref{cond2}) of Theorem \ref{Th2_clone}. Therefore, as is easy to check, the clone $\mathcal G$ contains an $m$-ary rule $f_1$ such that
$$
\widehat f_1(d_1(b_i), d_2(b_i), \ldots, d_m(b_i))=d(b_i) 
$$
for all $i$, $1\leq i\leq m$. So, $f_1(d_1, d_2, \ldots, d_m)=d$ which implies $D=D^+$.
\par The implication $(3)\Rightarrow (4)$ follows from the fact that the Boolean clones $O_1$, $D_1$, $D_2$ and $L_4$ are generated by the functions $x$, $\overline xy\vee \overline yz\vee yz$, $xy\vee yz\vee xz$, and $x\oplus y\oplus z$, respectively. From this it is easy to deduce that the $2$-clones $O_1^{\Uparrow}(A)$, $D_1^{\Uparrow}(A)$, $D_2^{\Uparrow}(A)$, and $L_4^{\Uparrow}(A)$ are generated by the functions $\widehat{\delta}$, $\widehat{\nu}$, $\widehat{\mu}$ and $\widehat{\lambda}$, respectively. This means that there are only four clones $\mathcal G_\delta$, $\mathcal G_\nu$, $\mathcal G_\mu$ and $\mathcal G_\lambda$ of local and neutral aggregation rules for $\mathfrak C_2(A)$, which are generated by the functions $\delta$, $\nu$, $\mu$ and $\lambda$, respectively. For each local and neutral aggregation rule $f$, we find the smallest (by inclusion) clone $\mathcal G_\chi$, $\chi\in \{\delta, \nu, \mu, \lambda\}$, that contains $f$. Then the clones generated by the rules $f$ and $\chi$ coincide. This means that $f$ and $\chi$ have the same invariants.
\par The implication $(4)\Rightarrow (1)$ follows from Proposition \ref{PropNeutrality}, which completes the proof. 

\section{Concluding remarks}

The universal algebraic approach offers hope for extending the results of Theorem \ref{Th1} to more general situations. However, such research will likely prove significantly more difficult. In particular, for aggregation rules on a set $C$ of choice functions $c: [A]^r\to A$, $r>2$, the analogue of Proposition \ref{Prop3} is not true (for this case, it is known that local aggregation rules do not have symmetric invariants $D\subseteq C$ if $5\leq |A|<\omega$, see \cite{Shelah,Polyakov1}). Another problem is the possible emergence of ``foreign'' combinatorics. For example, Theorem~\ref{Th2_clone} allows us to classify all logical local aggregation rules on symmetric functional sets $C\subseteq A^B$ with the condition $\forall b\in B\, |\{c(b): c\in C\}|\leq 2$ (in addition to the one considered, such a problem may arise, for example, if we study the aggregation of ``choice functions'' on multi- or fuzzy sets). In addition to neutral rules, we obtain here the rules given by Item \ref{cond1} of Theorem~\ref{Th2_clone}, which, apparently, are of little interest in social choice theory. This suggests that the generalization of the obtained results should be accompanied by additional restrictions on the aggregation rules or an expansion of the concept of triviality of a set $D\subseteq C$.

\end{document}